\documentclass[11pt]{article}

\usepackage[T1]{fontenc}
\usepackage[utf8]{inputenc}
\usepackage{amsmath,amssymb,amsfonts}
\usepackage{graphicx}
\usepackage{booktabs}
\usepackage{float}
\usepackage{hyperref}
\usepackage{cite}
\usepackage{textcomp}
\usepackage{microtype}
\usepackage{tikz}
\usetikzlibrary{arrows.meta,positioning,calc}

\usepackage[margin=1in]{geometry}

\usepackage{setspace}
\usepackage{amsthm}
\theoremstyle{plain}
\newtheorem{proposition}{Proposition}
\theoremstyle{remark}

\hypersetup{colorlinks=true,linkcolor=blue,citecolor=blue,urlcolor=blue}

\title{Explicit Interaction Architectures for Dynamical Learning: A Controlled Study of Structural Inductive Bias}

\author{Augusto Sarti
\thanks{Dipartimento di Elettronica, Informazione e Bioingegneria (DEIB), Politecnico di Milano, Piazza L. Da Vinci 32, 20133 Italy.}
}

\begin{document}

\maketitle

\begin{abstract}
We investigate a structure-first approach to dynamical learning in which the organization of stateful interactions is prescribed explicitly rather than left entirely to a generic recurrent parameterization. We introduce causal recurrent units built from an ordered sequence of local, state-modulated transformations. The construction is motivated by wave-based interaction models, but the units studied here do not impose scattering, passivity, or energy-balance constraints.

Because fixed recurrent dynamics, designed reservoir topologies, readout-only learning, and recurrent depth are already well established, the empirical question is deliberately narrower: does the proposed interaction organization provide a useful inductive bias under controlled computational conditions? We compare a one-layer structured model, a two-layer structured model, and a generic echo-state network (ESN), all with 12 recurrent states and the same strictly linear ridge readout. Each model family receives the same random-search budget on calibration data that are disjoint from the final test data, after which the selected hyperparameters are frozen.

On a custom nonlinear identification task, the one-layer structured model attains a mean validation NMSE of $2.76\times 10^{-4}$, compared with $3.19\times 10^{-4}$ for the two-layer model and $3.94\times 10^{-4}$ for the ESN. On NARMA10 the ordering reverses: the ESN attains $0.312$, compared with $0.348$ and $0.357$ for the one- and two-layer structured models. Thus, the proposed organization can be competitive and advantageous on one task, but it is not universally superior; moreover, recurrent depth does not provide a systematic benefit under matched state dimension. The results support a task-dependent interpretation of structural inductive bias and delimit the properties that can be attributed to the recurrent factorization itself, independently of stronger physical or system-theoretic constraints.
\end{abstract}

\section{Introduction}

Learning dynamical systems from data is a central problem in signal processing, control, and machine learning. Recurrent neural networks and related deep architectures can approximate complex temporal input--output relations through flexible nonlinear parameterizations \cite{elman1990,hochreiter1997,paszke2019}. A complementary line of work introduces stronger inductive biases by constraining the model class through physical or structural principles, including neural differential equations, physics-informed learning, Hamiltonian models, and graph-based relational architectures \cite{chen2018,raissi2019,greydanus2019,battaglia2018}.

A particularly relevant precedent is reservoir computing (RC), where the internal recurrent dynamics are commonly fixed and only a readout is trained. This literature has long shown that useful temporal representations need not arise from extensive adaptation of recurrent parameters. Moreover, deterministic and strongly constrained reservoir topologies can remain effective \cite{rodan2011}, while deep and hierarchical reservoir architectures exhibit layer-dependent dynamics, memory properties, and multiple effective time scales \cite{gallicchio2017esp,gallicchio2017deep,gallicchio2018design,manneschi2021}. The effect of connectivity and topology on the dynamical repertoire of reservoirs has also been studied explicitly \cite{dale2021}. In particular, constrained topologies have been investigated directly in DeepESNs, where structured recurrent organizations can interact beneficially with depth \cite{gallicchio2019topology}. Recent universality results for simple-cycle reservoirs further emphasize that highly restricted recurrent structures can retain substantial approximation power \cite{li2024}, while newer DeepESN variants also study explicit interactions among reservoirs \cite{shang2026}.

These precedents change the novelty question addressed in the present work. We do not claim that readout-only learning, fixed recurrent representations, or benefits from recurrent depth are new phenomena. Instead, we ask a narrower architectural question: whether a computational grammar built from explicitly ordered local interactions, motivated by wave-based modeling, can provide a useful constructive design bias for dynamical learning.

We introduce a class of \emph{structured dynamical units} with internal state and a prescribed sequence of local transformations. The construction is inspired by wave digital modeling \cite{sarti1999}. In nonlinear and multiply coupled wave-digital configurations, instantaneous interaction constraints can generate delay-free algebraic loops that require topology-dependent treatments, transformations, lookup methods, or iterative solution strategies \cite{sarti1999,medine2016}. Here, causal schedulability is imposed as an architectural constraint: local stages are ordered so that the current output and next state are obtained by a finite sequence of explicit operations, and units can be composed without creating new instantaneous algebraic dependencies between layers.

Learnable wave-digital models provide important prior art. Differentiable WDF implementations have been used for data-driven parameter and component modeling \cite{chowdhury2022}; recurrent neural models have been embedded in wave-digital blocks to represent dynamical hysteresis \cite{massi2023}; Lipschitz-bounded neural scattering maps have been proposed for multi-nonlinearity WDF configurations \cite{massi2024}; and neural nonlinear-element models have been combined with biparametric WDFs while retaining explicit computability \cite{longo2026}. Against this background, the contribution examined here is the abstraction of an ordered local-interaction grammar into a generic recurrent cell, outside a specific circuit-scattering realization, together with a controlled study of the inductive bias induced by that factorization.

We evaluate the architecture in a reservoir-style protocol on two nonlinear dynamical tasks. A one-layer structured model with 12 states, a two-layer model with two six-state units, and a 12-state generic ESN are compared using the same linear ridge readout. Hyperparameters are calibrated with equal random-search budgets on data disjoint from the final test sets and are then frozen. The design controls total recurrent state dimension, readout complexity, and calibration effort, while deliberately treating the ESN as a family-level reference rather than a graph-isomorphic ablation.

The interaction maps studied here are not constrained by scattering, passivity, or energy-balance relations. The present paper therefore isolates the recurrent factorization itself and establishes what can, and cannot, be attributed to that architectural choice alone.

The main contributions of this work are:
\begin{itemize}
\item A recurrent-cell factorization into an ordered sequence of state-modulated local maps, with linear parameter and per-sample update scaling in the total structured-state dimension.
\item A feedforward stacking construction in which such units can be layered without introducing instantaneous algebraic dependencies between layers.
\item A global boundedness result for the concrete saturating state update used in the experiments, together with an explicit uniform state bound under the initialization regime considered here.
\item A controlled comparison against a generic ESN under matched total state dimension, identical linear readout complexity, disjoint calibration/test data, and equal hyperparameter-search budgets.
\item Empirical evidence that the resulting structural bias is task dependent: the one-layer structured model outperforms the ESN on the custom nonlinear system, whereas the ESN performs better on NARMA10; the two-layer structured model does not systematically improve upon the one-layer model.
\item A positioning of the architecture relative to reservoir computing and learnable wave-digital modeling, with stronger physical properties such as passivity and energy balance kept explicitly outside the claims of the present study.
\end{itemize}

\section{Structured Dynamical Unit}

We introduce a computational unit characterized by local interactions, internal state, and nonlinear modulation. The unit is organized as a sequence of local interaction mappings motivated by wave-based structured computation, but implemented here as a fully explicit recurrent computation. The term \emph{wave-inspired} refers to this organizational motivation; the maps introduced below are not required in the present study to satisfy wave-digital scattering, passivity, or energy-balance constraints.

\begin{figure}[t]
\centering
\resizebox{0.98\linewidth}{!}{%
\begin{tikzpicture}[font=\small,node distance=0.72cm and 0.78cm,
block/.style={draw,rounded corners=3pt,thick,minimum width=1.25cm,minimum height=0.72cm,align=center,fill=gray!8},
state/.style={draw,rounded corners=3pt,thick,minimum width=1.55cm,minimum height=0.72cm,align=center,fill=gray!12},
arr/.style={-Latex,thick},modarr/.style={-Latex,thick,dashed,gray!65}]
\node (input) {$a_n$};
\node[block,right=of input] (r1) {$R_1$};
\node[right=of r1] (z1) {$z_{1,n}$};
\node[block,right=of z1] (r2) {$R_2$};
\node[right=of r2] (z2) {$z_{2,n}$};
\node[block,right=of z2] (r3) {$R_3$};
\node[right=of r3] (rn) {$r_n$};
\draw[arr] (input)--(r1); \draw[arr] (r1)--(z1); \draw[arr] (z1)--(r2);
\draw[arr] (r2)--(z2); \draw[arr] (z2)--(r3); \draw[arr] (r3)--(rn);
\node[state,below=1.55cm of r2] (state) {$\xi_n$};
\node[block,below=1.55cm of rn] (phi) {$\phi$};
\node[state,right=0.85cm of phi] (next) {$\xi_{n+1}$};
\node[block,right=0.75cm of next] (mix) {$m^\top$};
\node[right=0.65cm of mix] (out) {$b_n$};
\draw[modarr] (state.north west)--(r1.south);
\draw[modarr] (state.north)--(r2.south);
\draw[arr] (state)--(phi); \draw[arr] (rn)--(phi.north);
\draw[arr] (phi)--(next); \draw[arr] (next)--(mix); \draw[arr] (mix)--(out);
\end{tikzpicture}%
}
\caption{Explicit structured dynamical unit. Ordered local interactions generate an intermediate vector $r_n$, which enters the causal state update. The scalar layer output $b_n$ is projected from the updated state. All within-step dependencies admit a predetermined evaluation order.}
\label{fig:structured_unit}
\end{figure}
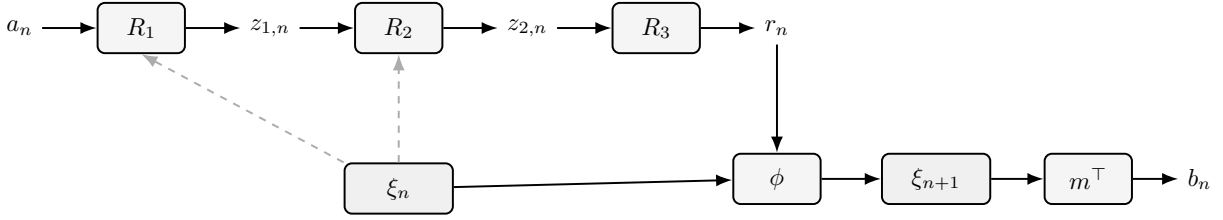

Let $a_n$ denote the scalar input at time $n$ and $\xi_n\in\mathbb{R}^{d}$ the internal state. The concrete realization used in the experiments is
\begin{align}
z_{1,n} &= \tanh\!\left(A_1\begin{bmatrix}a_n\\ \bar\xi_n\end{bmatrix}+d_1\right), \\
z_{2,n} &= \tanh\!\left(A_2\begin{bmatrix}a_n\\ \bar z_{1,n}\\ \bar\xi_n\end{bmatrix}+d_2\right), \\
r_n &= \tanh\!\left(A_3\begin{bmatrix}a_n\\ \bar z_{1,n}\\ \bar z_{2,n}\end{bmatrix}+d_3\right),
\end{align}
where $\bar v$ denotes the arithmetic mean of the components of $v$. The state update is
\begin{equation}
\xi_{n+1}=\alpha\odot\xi_n+\gamma\odot\tanh\!\left(r_n+\kappa\xi_n\right),
\label{eq:stateupdate}
\end{equation}
and the scalar signal passed to the next layer is
\begin{equation}
b_n=m^\top\xi_{n+1},\qquad \|m\|_2=1.
\label{eq:layeroutput}
\end{equation}
The matrices $A_1,A_2,A_3$, biases $d_1,d_2,d_3$, retention vector $\alpha$, injection vector $\gamma$, feedback coefficient $\kappa$, and projection $m$ are fixed after initialization in the experiments; only the final readout is fitted.

The construction is a directed recurrent computational graph with an exposed dependency pattern. It should not be interpreted as defining a function class known to lie outside generic recurrent neural architectures: sufficiently flexible recurrent cells can represent comparable input--state maps. Its role is instead to fix a particular factorization that can be studied as an architectural bias under controlled conditions.

All computations are explicit and sequential within a time step. This schedulability is useful for feedforward stacking, but it is not claimed as a unique property of the model; ordinary acyclic computational graphs can of course be explicitly evaluated as well. The present formulation should therefore be viewed as a wave-inspired recurrent construction rather than as a complete wave-digital realization.
\section{Layered Architecture}
\label{sec:layered}

Structured units can be composed into feedforward stacks while each layer retains its own recurrent state. For a stack of depth $L$, let $a_n^{(1)}=a_n$ and let
\begin{equation}
a_n^{(\ell+1)}=b_n^{(\ell)}=m_\ell^\top\xi_{n+1}^{(\ell)},\qquad \ell=1,\ldots,L-1,
\end{equation}
where each state $\xi_n^{(\ell)}$ is updated according to Eq.~\eqref{eq:stateupdate} using its layer input $a_n^{(\ell)}$. The features exposed to the readout are the concatenated updated states
\begin{equation}
X_n=\begin{bmatrix}(\xi_{n+1}^{(1)})^\top & \cdots & (\xi_{n+1}^{(L)})^\top\end{bmatrix}^\top.
\end{equation}
In the experiments the prediction is a strictly linear function of the current external input and these fixed dynamical features,
\begin{equation}
\widehat y_n=w_0+w_u a_n+w_X^\top X_n.
\label{eq:readout}
\end{equation}

The layers do not share state variables. Each additional layer changes the factorization and the distribution of a fixed total state budget rather than merely adding a feedforward transformation. Hierarchically organized recurrent dynamics are not unique to the present architecture: DeepESNs and related models already establish that recurrent depth can alter memory and temporal representation \cite{gallicchio2017esp,gallicchio2017deep,gallicchio2018design,manneschi2021}. Accordingly, the controlled experiment below explicitly tests whether depth is beneficial once total state dimension is held fixed.
\section{Relation to Reservoir Computing and Structured Recurrent Models}

The proposed framework is closely related to reservoir computing because its internal dynamical state can be held fixed while a linear readout is trained. In classical echo-state networks the reservoir is often randomly initialized, but randomness is not essential to the RC paradigm: deterministic, sparse, cyclic, and otherwise designed reservoir topologies have been studied extensively \cite{rodan2011,dale2021,li2024}. Thus, the mere use of fixed internal dynamics or a prescribed topology is not a novelty claim of the present work.

The relation becomes even closer for layered models. Deep reservoir computing and DeepESNs explicitly study hierarchical compositions of recurrent layers, including their stability, memory, frequency response, and representation of multiple time scales \cite{gallicchio2017esp,gallicchio2017deep,gallicchio2018design,manneschi2021}. Consequently, the readout-only experiments below should be interpreted as a diagnostic of the features generated by the proposed units, rather than as evidence that readout-only learning or hierarchical temporal representation is new.

The intended object of study lies instead in the \emph{local interaction grammar}. Each unit is decomposed into an ordered sequence of state-modulated local transformations. Such explicit ordering is not unique to this architecture---ordinary recurrent cells and computational graphs can also be explicitly evaluated---so schedulability alone is not a novelty claim. The question is whether fixing this particular dependency pattern provides a useful constructive scaffold for layered recurrent modeling.

\begin{table}[t]
\centering
\small
\caption{Conceptual positioning relative to reservoir models. The entries describe typical or defining properties rather than exhaustive restrictions on each model family.}
\label{tab:positioning}
\begin{tabular}{p{0.22\linewidth}p{0.21\linewidth}p{0.21\linewidth}p{0.22\linewidth}}
\toprule
Property & Reservoir / ESN & Deep reservoir / DeepESN & Present architecture \\
\midrule
Trainable readout with fixed internal dynamics & Standard option & Standard option & Used in the experiments \\
Designed internal topology & Possible & Possible & Prescribed local interaction grammar \\
Multiple recurrent layers & Not required & Defining feature & Supported by stacking \\
Layer-dependent temporal dynamics & Possible & Well established & Possible; not claimed as unique \\
Within-cell factorization into ordered local maps & Architecture dependent & Architecture dependent & Imposed by construction \\
Wave-based physical interpretation & Not required & Not required & Design motivation only in this paper \\
Passivity / energy guarantee & Not generic & Not generic & Not imposed here \\
\bottomrule
\end{tabular}
\end{table}

The experiments below improve on an unmatched depth comparison by holding total state dimension and readout complexity fixed and by giving each model family the same calibration budget. The generic ESN provides a reference recurrent organization rather than a topology-only ablation: internal coefficient counts and graph structures are not identical. The results therefore support statements about comparative task-dependent bias under controlled state/readout/calibration budgets, but they do not identify a single causal graph edge or interaction as the source of a performance difference.

\paragraph{Parameter and update scaling.}
The local-interaction factorization also imposes a specific scaling with state dimension. For a layer of dimension $d$, the realization in Eqs.~(1)--(5) contains $2d+3d+3d$ matrix coefficients in $A_1,A_2,A_3$, $3d$ bias coefficients, and $d$ coefficients each in $\alpha$, $\gamma$, and $m$. In the implementation used here, the feedback coefficient $\kappa$ is shared across layers; hence a structured model with total state dimension $D$ contains $14D+1$ fixed internal scalar coefficients, independently of how $D$ is distributed across layers. Its per-sample state update therefore requires $O(D)$ arithmetic because each local affine map has fixed input dimension. By comparison, a conventional $D$-state ESN uses a $D\times D$ recurrent matrix, with $O(D^2)$ arithmetic in a dense implementation or $O(\mathrm{nnz}(W))$ when sparse structure is exploited. This scaling difference is not used as a matching criterion in the experiments and is not unique among constrained reservoirs; it is recorded here as a structural consequence of the proposed factorization.

\section{Computational Dependency and State Boundedness}

The dependency pattern of a structured unit is acyclic within each time step. Starting from known $a_n$ and $\xi_n$, the variables are obtained in the order
\[
(a_n,\xi_n)\longrightarrow z_{1,n}\longrightarrow z_{2,n}\longrightarrow r_n
\longrightarrow \xi_{n+1}\longrightarrow b_n.
\]
A feedforward stack remains explicitly schedulable because the output of layer $\ell$ is available before layer $\ell+1$ is evaluated. These are computational dependency properties; they do not by themselves imply passivity or an echo-state property. For the concrete state update used in the experiments, however, global state boundedness follows directly from the saturating injection.

\begin{proposition}[Global boundedness of the structured state]
Consider the state update in Eq.~\eqref{eq:stateupdate} and define $\bar\alpha=\|\alpha\|_\infty$. If $\bar\alpha<1$, then for every input sequence for which the local interaction maps are well defined and every initial state,
\begin{equation}
\|\xi_n\|_\infty
\leq
\bar\alpha^n\|\xi_0\|_\infty
+
\|\gamma\|_\infty\frac{1-\bar\alpha^n}{1-\bar\alpha}.
\label{eq:statebound}
\end{equation}
Consequently,
\begin{equation}
\limsup_{n\to\infty}\|\xi_n\|_\infty
\leq \frac{\|\gamma\|_\infty}{1-\bar\alpha},
\end{equation}
and the layer output $b_n=m^\top\xi_{n+1}$ is bounded as well.
\end{proposition}

\begin{proof}
Because $|\tanh(v)|\leq 1$ componentwise, Eq.~\eqref{eq:stateupdate} gives
\[
|\xi_{i,n+1}|\leq |\alpha_i|\,|\xi_{i,n}|+|\gamma_i|.
\]
Taking the infinity norm yields
\[
\|\xi_{n+1}\|_\infty
\leq \bar\alpha\|\xi_n\|_\infty+\|\gamma\|_\infty.
\]
Iterating this scalar recursion gives Eq.~\eqref{eq:statebound}. Finally, $\|m\|_2=1$ implies $|b_n|\leq\|\xi_{n+1}\|_2\leq\sqrt{d}\,\|\xi_{n+1}\|_\infty$.
\end{proof}

The experimental hyperparameter search enforces $\alpha_i\leq 0.97$ for every structured realization, so the proposition applies to all structured models reported below. For this specific implementation the bound is independent of input amplitude: saturation of the state injection prevents the input from entering the state-growth term directly. It is nevertheless an absolute boundedness statement, not an incremental contraction result. Uniqueness of the driven state or an echo-state property would require a bound on differences between trajectories; related contractivity conditions for layered reservoirs are developed in \cite{gallicchio2017esp}. Passivity or energy non-expansiveness likewise requires additional structure not imposed here.

\section{Experimental Setup}
\label{sec:experimental}

The experiments are designed to test a narrower question than state-of-the-art prediction accuracy: whether the proposed recurrent factorization behaves as a useful inductive bias when total state dimension, readout complexity, and calibration effort are controlled.

\subsection{Tasks and data}
Two scalar-input/scalar-output tasks are used. The first is a custom nonlinear dynamical system. Its input combines two sinusoids, additive Gaussian noise, and eight exponentially decaying signed impulses. The input is normalized by the maximum absolute value computed on the training segment only. The target is generated recursively as
\begin{align}
y_n={}&0.72y_{n-1}-0.18y_{n-2}+0.24\tanh(1.8u_n) \\
&+0.10u_nu_{n-1}-0.08y_{n-1}u_n+0.05\tanh(2y_{n-1}).
\label{eq:customtarget}
\end{align}
The second task is NARMA10, generated from $u_n\sim\mathcal{U}[0,0.5]$ according to
\begin{equation}
y_{n+1}=0.3y_n+0.05y_n\sum_{i=0}^{9}y_{n-i}+1.5u_{n-9}u_n+0.1.
\label{eq:narma}
\end{equation}
Each realization contains 5000 samples, of which the first 3000 are used for readout fitting and the remaining 2000 for validation. The first 200 samples of each train/validation run are treated as washout and excluded from the error metric.

\subsection{Models and readout}
Three models are compared:
\begin{itemize}
\item \textbf{Structured-1}: one structured unit with 12 recurrent states;
\item \textbf{Structured-2}: two structured units with six recurrent states each, for 12 recurrent states in total;
\item \textbf{Generic-ESN}: a leaky echo-state network with 12 recurrent states.
\end{itemize}
Thus Structured-1 and Structured-2 have the same total state dimension; with the parameterization of Eqs.~(1)--(5), they also contain the same number of structured-unit scalar coefficients before the readout. The ESN is not coefficient-count matched because its recurrent matrix has a different sparse connectivity model. All three architectures expose 12 recurrent features to exactly the same readout in Eq.~\eqref{eq:readout}.

The readout is fitted by ridge regression with $\lambda=10^{-6}$:
\begin{equation}
\widehat w=\arg\min_w\|\Phi w-y\|_2^2+\lambda\|R^{1/2}w\|_2^2,
\end{equation}
where $\Phi=[\mathbf{1},u,X]$ and the intercept is not regularized. No quadratic state lifting or internal-parameter training is used in the reported comparison. The primary metric is normalized mean-square error (NMSE), $\mathrm{MSE}/\operatorname{var}(y)$, evaluated after washout.

\subsection{Equal-budget calibration and frozen test}
Calibration and test data are disjoint. For each task, three independent calibration datasets and ten model initializations per dataset are used. Each model family receives 40 random hyperparameter trials. All candidates within a family are evaluated using common base random draws, reducing ranking noise caused by initialization. The setting with the lowest median calibration NMSE is selected separately for each task and model family.

The structured search varies the common matrix scale, the lower and upper ranges for state retention $\alpha$ and injection $\gamma$, the feedback coefficient $\kappa$, and the bias scale. The ESN search varies spectral radius, input scale, bias scale, recurrent density, and leak. After calibration, the selected setting is frozen. Final evaluation uses five new independently generated test datasets and 30 independent model initializations per dataset, giving 150 test trials for each model and task. The supplied MATLAB script records all deterministic seeds and sampled hyperparameter settings.

\begin{table}[H]
\centering
\small
\caption{Hyperparameter settings selected on the disjoint calibration sets and frozen for final testing. Ranges denote the intervals from which componentwise $\alpha$ and $\gamma$ values are sampled at initialization.}
\label{tab:selected_hyperparameters}
\begin{tabular}{llccccc}
\toprule
Task & Model & $A$ scale & $\alpha$ range & $\gamma$ range & $\kappa$ & Bias scale \\
\midrule
Custom & Structured-1 & 0.636 & [0.575,0.798] & [0.160,0.239] & 0.176 & 0.050 \\
Custom & Structured-2 & 0.633 & [0.505,0.616] & [0.156,0.283] & 0.243 & 0.085 \\
NARMA10 & Structured-1 & 0.691 & [0.493,0.693] & [0.331,0.430] & 0.530 & 0.134 \\
NARMA10 & Structured-2 & 0.606 & [0.637,0.802] & [0.319,0.375] & 0.386 & 0.159 \\
\bottomrule
\end{tabular}

\vspace{0.6em}
\begin{tabular}{llccccc}
\toprule
Task & Model & Spectral radius & Input scale & Bias scale & Density & Leak \\
\midrule
Custom & Generic-ESN & 0.277 & 0.598 & 0.196 & 0.284 & 0.750 \\
NARMA10 & Generic-ESN & 1.101 & 0.137 & 0.031 & 0.525 & 0.523 \\
\bottomrule
\end{tabular}
\end{table}

This protocol matches total recurrent state dimension, readout form, training data length, and hyperparameter-search budget. It does not make the ESN and structured units graph-isomorphic or equalize every internal coefficient, so the comparison should be interpreted as a controlled family-level comparison rather than a topology-only causal ablation.

\section{Results}

\subsection{Aggregate validation performance}
Table~\ref{tab:controlled_results} reports the final frozen-hyperparameter results. Figures~\ref{fig:custom_summary} and \ref{fig:narma_summary} show the corresponding mean validation NMSE and one standard deviation across the 150 test trials.

\begin{table}[t]
\centering
\small
\caption{Validation NMSE after equal-budget calibration and frozen-hyperparameter testing. Each entry summarizes 5 independent test datasets $\times$ 30 model initializations.}
\label{tab:controlled_results}
\begin{tabular}{llcc}
\toprule
Task & Model & Mean $\pm$ std. & Median [IQR] \\
\midrule
Custom & Structured-1 & $(2.761\pm0.374)\!\times\!10^{-4}$ & $2.698\,[2.500,2.917]\!\times\!10^{-4}$ \\
Custom & Structured-2 & $(3.187\pm1.034)\!\times\!10^{-4}$ & $2.954\,[2.582,3.533]\!\times\!10^{-4}$ \\
Custom & Generic-ESN  & $(3.940\pm2.563)\!\times\!10^{-4}$ & $3.516\,[2.256,4.833]\!\times\!10^{-4}$ \\
\midrule
NARMA10 & Structured-1 & $0.348\pm0.035$ & $0.360\,[0.324,0.371]$ \\
NARMA10 & Structured-2 & $0.357\pm0.042$ & $0.364\,[0.323,0.389]$ \\
NARMA10 & Generic-ESN  & $0.312\pm0.077$ & $0.293\,[0.274,0.321]$ \\
\bottomrule
\end{tabular}
\end{table}

\begin{figure}[t]
\centering
\includegraphics[width=0.76\linewidth]{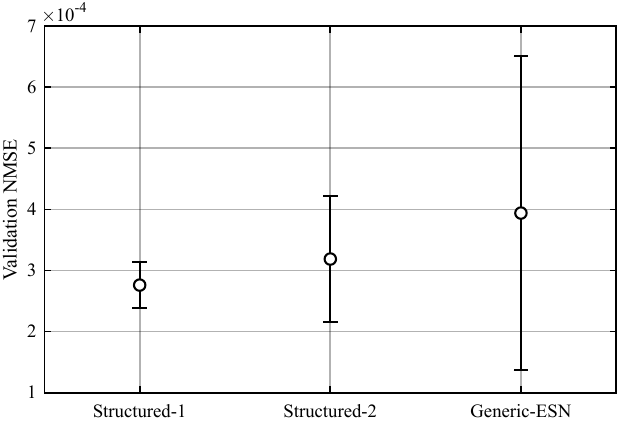}
\caption{Custom nonlinear system: validation NMSE after calibration on disjoint data and frozen-hyperparameter testing. Markers show the mean and error bars one standard deviation over 150 trials. All models use 12 recurrent states and the same linear ridge readout.}
\label{fig:custom_summary}
\end{figure}

For the custom nonlinear system, Structured-1 achieves the lowest mean NMSE, approximately 30\% below that of the generic ESN. Structured-2 is also below the ESN mean, but it is approximately 15\% worse than Structured-1. The structured models also exhibit lower trial-to-trial dispersion in this task. The result therefore supports the existence of a useful task-specific structural bias, but it does not support a benefit from depth under a fixed total state dimension.

\begin{figure}[t]
\centering
\includegraphics[width=0.76\linewidth]{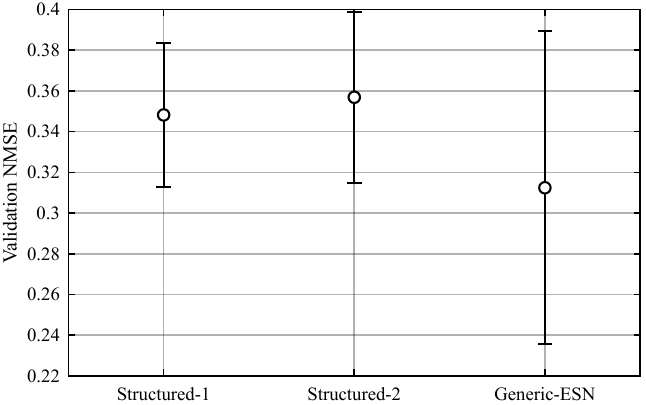}
\caption{NARMA10: validation NMSE under the same controlled protocol as Fig.~\ref{fig:custom_summary}. The generic ESN has the lowest average and median error, while the two structured depths are comparable to each other.}
\label{fig:narma_summary}
\end{figure}

On NARMA10, the ordering reverses. The generic ESN obtains a mean NMSE of 0.312, about 10\% below Structured-1 and 12\% below Structured-2. Structured-2 again does not improve on Structured-1. The ESN has greater variability across initializations, but its median remains clearly lower. Taken together, the two tasks do not support a universal superiority claim for either recurrent organization.

\subsection{Representative post-washout trajectories}
Figures~\ref{fig:custom_trace} and \ref{fig:narma_trace} provide qualitative views from the fixed representative trial defined deterministically by the first test dataset and first initialization. These plots are illustrative only; all quantitative conclusions use the full 150-trial distributions summarized above.

\begin{figure}[t]
\centering
\includegraphics[width=0.88\linewidth]{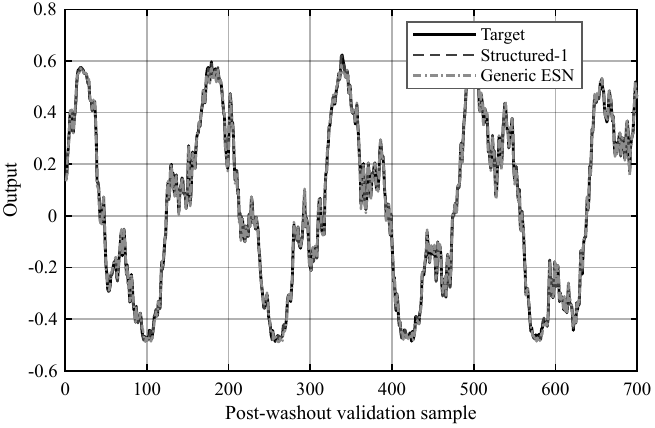}
\caption{Custom nonlinear system: representative post-washout validation trajectory for Structured-1 and the generic ESN. The trial is fixed by dataset index 1 and initialization index 1 rather than selected by performance.}
\label{fig:custom_trace}
\end{figure}

For the custom system, both models visually track the target closely over the displayed segment, so the aggregate statistics in Fig.~\ref{fig:custom_summary} are more informative than visual inspection alone. The trajectory nevertheless shows that a fixed structured recurrent feature map followed by a linear readout can reproduce the principal temporal variations of the target.

\begin{figure}[t]
\centering
\includegraphics[width=0.88\linewidth]{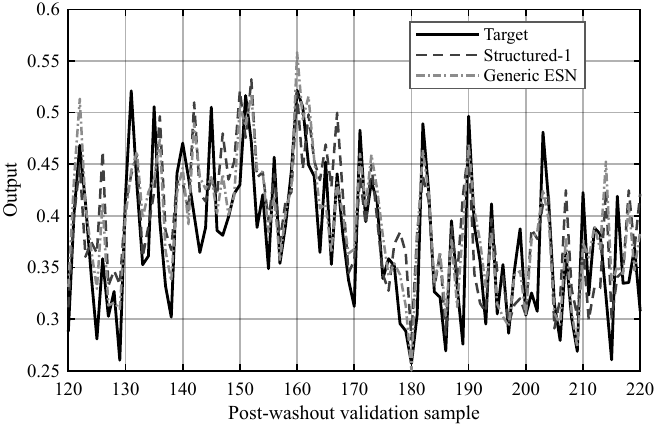}
\caption{NARMA10: detail from the same fixed representative-trial protocol. For readability, an arbitrarily selected post-washout interval (samples 120--220) is displayed; the interval was not chosen to optimize the relative appearance of either model.}
\label{fig:narma_trace}
\end{figure}

The NARMA10 detail makes the model differences easier to see than a 700-sample overview. Neither prediction reproduces the target perfectly: both smooth or miss some rapid excursions, and the aggregate NMSE is therefore essential for judging performance. The zoom is used only for readability and does not enter model selection or quantitative evaluation.

\subsection{Numerical conditioning}
The fitted readouts can be numerically ill-conditioned despite ridge regularization. Median condition numbers of the regularized normal matrices range from approximately $3\times10^8$ to $1\times10^{10}$ across the six task/model combinations. Because the same fixed $\lambda=10^{-6}$ and the same readout construction are used for all models, the controlled comparison remains internally consistent, but absolute performance values can be sensitive to feature scaling and regularization. As a robustness check, we repeated the frozen-dynamics evaluation for $\lambda\in\{10^{-7},10^{-6},10^{-5},10^{-4},10^{-3}\}$ without retuning any recurrent hyperparameters. The qualitative family ordering was unchanged throughout the sweep: Structured-1 retained the lowest mean NMSE on the custom task, while the generic ESN retained the lowest mean NMSE on NARMA10. Thus, the main comparative conclusions do not depend on the particular choice $\lambda=10^{-6}$, although absolute errors and readout norms vary with regularization.
\section{Discussion}

The controlled experiment supports a specific interpretation of the proposed factorization: it acts as a task-dependent inductive bias rather than as a universally superior recurrent organization. On the custom nonlinear system Structured-1 is more accurate on average and less variable than the generic ESN, whereas on NARMA10 the ESN is more accurate. The regularization sweep preserves this reversal, so it is not an artifact of the nominal ridge coefficient.

Depth is likewise not a general explanation of performance. Structured-1 and Structured-2 use the same total number of recurrent states, the same readout, and the same calibration budget, yet Structured-2 is worse on average in both reported tasks. This is compatible with the reservoir literature, where the effect of depth depends on operating regime, topology, task, and the allocation of memory across layers \cite{gallicchio2017esp,gallicchio2017deep,gallicchio2018design,manneschi2021,gallicchio2019topology}.

The ESN comparison remains a family-level comparison rather than a topology-only causal ablation: matching state dimension, readout, data, and search effort does not make the two recurrent graphs isomorphic or equalize all internal coefficients. The observed differences therefore characterize the model families under the stated resource controls, not the causal effect of any single interaction edge.

The relation to wave digital modeling is correspondingly focused. The present maps adopt a local-interaction viewpoint and explicit scheduling, while scattering relations, passivity, and energy conservation are not imposed. Learnable WDFs already demonstrate that data-driven components can coexist with physically organized wave-digital models \cite{chowdhury2022,massi2023,massi2024,longo2026}. The role of the present study is to isolate the ordered recurrent grammar itself. Its mixed empirical outcome is informative: explicit organization alone can favor particular input--output dependencies, while any stronger physical or system-theoretic claims would require additional constraints not studied here.
\section{Conclusion}

We introduced an explicitly schedulable recurrent interaction unit and studied it as a fixed dynamical feature generator. The architecture has linear parameter and update scaling with total structured-state dimension, and its saturating state update admits a global boundedness guarantee under the initialization regime used here. Experimentally, the comparison controls total recurrent state dimension, readout complexity, calibration effort, and train/test separation.

The results support a deliberately limited conclusion: the prescribed recurrent organization behaves as a task-dependent inductive bias. Structured-1 outperforms the calibrated ESN on the custom nonlinear identification task, while the ESN performs better on NARMA10; allocating the same 12-state budget across two structured layers does not improve either task. The present study therefore characterizes what follows from this specific recurrent factorization--explicit local interaction, feedforward stacking, bounded state evolution, and controlled empirical behavior--while scattering, passivity, energy constraints, and stronger system-theoretic guarantees remain outside its scope.
\section*{Acknowledgment}
The author used generative AI tools for language editing and assistance in organizing the presentation of the manuscript. The author reviewed and verified all mathematical derivations, literature references, experimental results, and scientific conclusions and takes full responsibility for the content.

\bibliographystyle{IEEEtran}
\bibliography{references_v3}

\end{document}